\documentclass{article}

\usepackage{theorem}
\usepackage{amsmath,mathtools,amssymb,graphicx,float}

\def\BBox{\rule{2mm}{3mm}}
\def\QED{\hfill$\BBox$}
\newenvironment{proof}
{\begin{rm}\par\smallskip\noindent{\bf Proof.}\quad}{\QED\end{rm}}

\newenvironment{proof_main}
{\begin{rm}\par\smallskip\noindent{\bf Proof of Theorem~\ref{thm:main}.}\quad}{\QED\end{rm}}

\theorembodyfont{\rmfamily}
\newtheorem{thm}{Theorem}[section]
\newtheorem{lem}[thm]{Lemma}        %
\newtheorem{cor}[thm]{\bfseries Corollary}

\newtheorem{cl}{\bfseries Claim}

\title{A new upper bound for angular resolution}

\author{
Hiroyuki Miyata\\
Faculty of Informatics, Gunma University, Japan\\
\texttt{hmiyata@gunma-u.ac.jp}
}

\begin{document}
\maketitle

\begin{abstract}
The angular resolution of a planar straight-line drawing of a graph is the smallest angle formed by two edges incident to the same vertex.
Garg and Tamassia (ESA '94) constructed a family of planar graphs with maximum degree~$d$ that have angular resolution $O((\log d)^{\frac{1}{2}}/d^{\frac{3}{2}})$ in any planar straight-line drawing.
This upper bound has been the best known upper bound on angular resolution for a long time.
In this paper,  we improve this upper bound. 
For an arbitrarily small positive constant $\varepsilon$,  we construct a family of planar graphs  with maximum degree~$d$ that have angular resolution $O((\log d)^\varepsilon/d^{\frac{3}{2}})$ in any planar straight-line drawing.
\end{abstract}

\section{Introduction}
There are many ways to represent (draw) the same graph.
Typically, vertices are represented as points in the plane, and edges are  represented as line segments. 
Such a representation of a graph is called a \emph{straight-line drawing}.
A straight-line drawing is \emph{planar} if no pair of edges crosses.
In the field of graph drawing, extensive research has been made to find ``good'' drawings of graphs with respect to various  aesthetic criteria (see \cite{DETT99,NR04,T13}).
One of the important aesthetic criteria for graph drawing is the \emph{angular resolution}, i.e., the smallest angle formed by two edges incident to the same vertex.

The trivial upper bound on the angular resolution of a graph with maximum degree~$d$ is $2\pi/d$.
For planar straight-line drawings, non-trivial upper and lower bounds have been obtained by several researchers.
Malitz and Papakostas~\cite{MP94} showed that every planar graph with maximum degree $d$ has a planar straight-line drawing with angular resolution~$\Omega (1/7^d)$. 
On the other hand, Garg and Tamassia~\cite{GT94} provided a non-trivial upper bound by constructing a family of planar graphs with angular resolution $O((\log d)^{\frac{1}{2}}/d^{\frac{3}{2}})$ in any  planar straight-line drawing.
These upper and lower bounds have been the best known bounds on angular resolution, and there is a huge gap between the upper bound and the lower bound.

For some restricted graph classes, more results are known. 
Malitz and Papakostas~\cite{MP94} showed that every outerplanar graph with maximum degree $d$ can be drawn with angular resolution $\pi/2(d-1)$.
This bound was later improved into $\pi/(d-1)$ by Garg and Tamassia~\cite{GT94}.
The authors of  the paper~\cite{GT94} also proved that every planar 3-tree and series-parallel graph with maximum degree~$d$ has a  planar straight-line drawing with angular resolution $\Omega (1/d^2)$.
The bound for series-parallel graphs was improved by Lenhart et al.~\cite{LLMN23}.
That is, Lenhart etl al.~\cite{LLMN23} proved that every partial $2$-tree, a superclass of series-parallel graphs, with maximum degree $d$ has a  planar straight-line drawing with angular resolution $\pi/2d$.

Angular resolution has also been studied for other drawing conventions.
Formann et al.~\cite{FHHKLSWW90} proved that planar graph with maximum degree $d$ has a (not necessarily planar) straight-line drawing with angular resolution $\Omega (1/d)$.
Kant~\cite{K96} presented an algorithm to construct a polyline drawing of a planar graph with angular resolution at least $4/(3d+1)$ (and with small number of bends and small area).
Gutwenger and Mutzel~\cite{GM98} improved this bound into $2/d$.
Cheng et al.~\cite{CDGK01} presented an algorithm to construct a drawing of a planar graph  with angular resolution $\Omega (1/d)$ using at most two circular arcs per edge.

In this paper, we improve the upper bound on angular resolution for planar straight-line drawings.
In Section~\ref{sec:def}, we give some necessary definitions.
In Section~\ref{sec:lemma}, we review a geometric lemma presented in \cite{GT94}, which plays an important role in deriving the new upper bound.
In Section~\ref{sec:result}, we first construct a family of planar graphs with angular resolution $O((\log d)^{\frac{1}{6}}/d^{\frac{3}{2}})$.
Then, slightly extending the construction, we present a family of planar graphs with angular resolution $O((\log d)^\varepsilon/d^{\frac{3}{2}})$ for any $\varepsilon > 0$.

\section{Notations and definitions}
\label{sec:def}
We summarize some basic definitions and notations in graph drawing. For more details, see~\cite{DETT99,NR04}.
Let $G=(V,E)$ be a graph.
For $u,v \in V$, we denote by $uv$ the edge connecting vertices $u$ and $v$.
For $k \geq 1$,  $G$ is \emph{$k$-connected} if $G$ has more than $k$ vertices, and removal of any $k-1$ vertices of $G$ does not disconnect $G$.
A \emph{straight-line drawing} of $G$ is a geometric representation of $G$ where the vertices of $G$ are represented as distinct points in the plane and the edges of $G$ are represented as line segments.
A straight-line drawing is \emph{planar} if no pair of edges crosses.
In the rest of the paper, we always consider planar straight-line drawings, which are hereinafter simply referred to as \emph{drawings}.
A graph is said to be \emph{planar} if it has a planar drawing. 
The \emph{angular resolution} of a drawing is the smallest angle formed by two edges incident to the same vertex.
A drawing of a planar graph $G$ divides the plane into regions called \emph{faces}.
The unbounded face is called the \emph{outer face}. 
An \emph{embedding} of  $G$ is an equivalence class of drawings of $G$ determined by the collection of circular clockwise orderings of the edges incident to each vertex.
It is known that every $3$-connected planar graph has a unique embedding in the plane (see \cite[Chapter~2]{NR04}).
A \emph{planar $3$-tree} is a graph obtained from the triangle graph by successively adding a new vertex $v$ in some face $f$ and connecting $v$ to the three vertices of $f$.
Since a planar $3$-tree is $3$-connected, it has a unique embedding.

For points $A,B,C$ in the plane, we denote by $\triangle{ABC}$ the triangle formed by $A$, $B$ and $C$.
By $\angle{ABC}$, we denote either the angle between line segments $AB$ and $BC$ or its measure.

\section{Geometric lemma}
\label{sec:lemma}
In this section, we review a geometric lemma presented by Garg and Tamassia~\cite{GT94}, which plays an important role in proving the upper bound on angular resolution.
Since a proof of the lemma is omitted in \cite{GT94}, we include a proof for completeness.
The lemma is described as follows:
\begin{lem}(\cite{GT94})
Let $\triangle{ABC}$ be a triangle with $\angle{BAC} \leq \pi/2$, 
$D$ a point inside $\triangle{ABC}$, and $\alpha_1$, $\alpha_2$, $\beta_1$, $\beta_2$, $\gamma_1$, $\gamma_2$ the angles defined as in Figure~\ref{fig:triangle_lemma}.
If $\alpha_2 \geq \alpha_1$, we have
\[ \min \left\{ \frac{\beta_2}{\beta_1}, \frac{\gamma_2}{\gamma_1} \right\} \leq \frac{\pi^2}{4} \sqrt{\frac{\alpha_1}{\alpha_2}}. \]
\label{lem:geometric}
\end{lem}
\begin{figure}[h]
\begin{center}
\includegraphics[scale=0.18, bb =  0 0 601 511,clip]{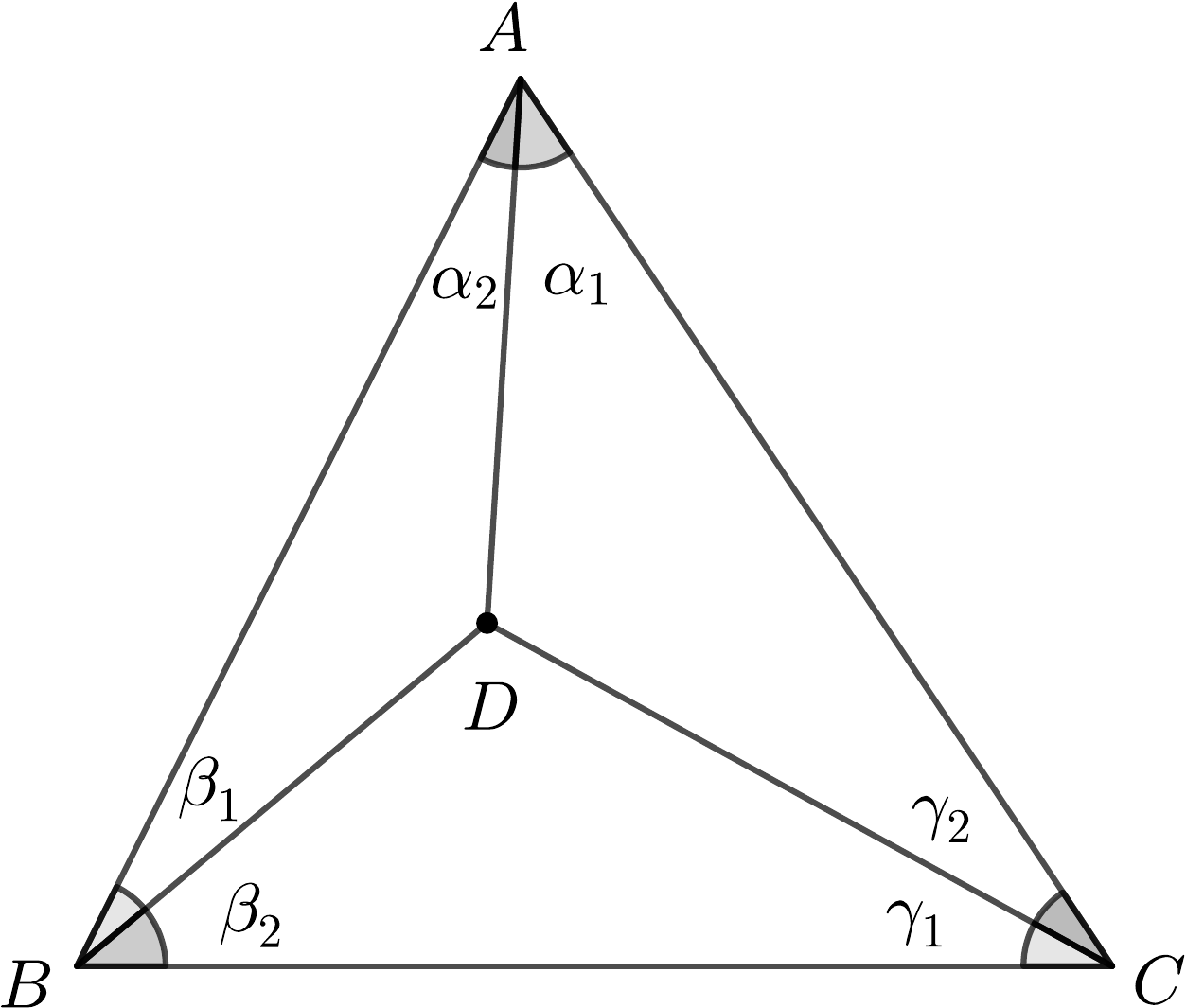}
\end{center}
\caption{$\triangle{ABC}$ in Lemma~\ref{fig:triangle_lemma}}
\label{fig:triangle_lemma}
\end{figure}
\mbox{}
\\

\begin{proof}
Without loss of generality, we assume that $\beta_2/\beta_1 \leq \gamma_2/\gamma_1$.
By repeated application of the law of sines, we obtain
\[ \frac{\sin \alpha_2}{\sin \alpha_1} \cdot \frac{\sin \beta_2}{\sin \beta_1} \cdot \frac{\sin \gamma_2}{\sin \gamma_1} = 1.\]
We first assume that $\beta_2, \gamma_2 \leq \pi/2$.
Then, since $\frac{2}{\pi}x \leq \sin x$ for $0 \leq x \leq \frac{\pi}{2}$ and $\sin x \leq x$ for $x \geq 0$, we have
\[ \frac{8}{\pi^3} \cdot \frac{\alpha_2}{\alpha_1} \cdot \frac{\beta_2}{\beta_1} \cdot \frac{\gamma_2}{\gamma_1} \leq 1.\]
Therefore, we have
\[ \left( \frac{\beta_2}{\beta_1}\right)^2 \leq \frac{\pi^3}{8} \cdot \frac{\alpha_1}{\alpha_2}.\]
This implies
\[ \frac{\beta_2}{\beta_1} \leq \frac{\pi^2}{4} \sqrt{\frac{\alpha_1}{\alpha_2}}. \]
Next, we assume that $\beta_2 \geq \pi/2$. Then, we have  $\beta_2 \geq \beta_1$ and $\gamma_2 \leq \pi/2$.
Since $\beta_1 + \beta_2 \leq \pi$, we have $\sin \beta_2 = \sin (\pi -\beta_2) \geq \sin \beta_1$.
Therefore, we have
\[ \frac{4}{\pi^2} \cdot \frac{\alpha_2}{\alpha_1} \cdot 1 \cdot \frac{\gamma_2}{\gamma_1} \leq 1.\] 
Since $\alpha_1/\alpha_2 \leq 1$, we have
\[  \frac{\beta_2}{\beta_1} \leq \frac{\gamma_2}{\gamma_1} \leq \frac{\pi^2}{4} \cdot \frac{\alpha_1}{\alpha_2} \leq \frac{\pi^2}{4} \sqrt{\frac{\alpha_1}{\alpha_2}}. \]
In a similar way, we can prove the inequality in the case that  $\gamma_2 \geq \pi/2$. 
\end{proof}

\section{A family of planar graphs with angular resolution $O((\log d)^{\varepsilon}/d^{\frac{3}{2}})$}
\label{sec:result}
In this section, we construct a family of planar graphs with maximum degree $d$
that have angular resolution $O((\log d)^{\varepsilon}/d^{\frac{3}{2}})$ in any drawing, for any $\varepsilon > 0$.

\subsection{Definition of graphs $\widetilde{H}_d^{(c)}$}
For $d \geq 1$, we first define a family of graphs $F_d$ as follows.
Let $F_1$ be the triangle graph with vertices $u_1$, $v_1$, and $w$.
For $d \geq 2$, let $F_d$ be the graph obtained by adding new vertices $u_d$ and $v_d$, and new edges $wu_d$, $wv_d$, $u_dv_d$, $u_du_{d-1}$, 
$u_dv_{d-1}$ to  graph $F_{d-1}$ (see Figure~\ref{fig:f_d}).
We call graph $F_d$ the \emph{$d$-frame graph}.
The vertex $w$ is called the \emph{root vertex} of $F_d$.
The path $u_d,\dots,u_1,v_1,\dots,v_d$ is called the \emph{central path} of $F_d$.
The maximum degree of $F_d$ is $2d$.

Next, we define a family of graphs $G^{(c)}_d$ for $c,d \geq 1$ as follows.
For $d \geq 1$, let $G^{(1)}_d$ be the ($d+1$)-frame graph with vertices $w^{(1)},  u^{(1)}_1,\dots,u^{(1)}_{d+1}, v^{(1)}_1,\dots,v^{(1)}_{d+1}$, where
$w^{(1)}$ is the root vertex,  and the central path is $u^{(1)}_{d+1},\dots,u^{(1)}_1,v^{(1)}_1,\dots,v^{(1)}_{d+1}$.
For $d \geq 1$ and $c \geq 2$, we define graph $G^{(c)}_d$ as follows.
We first consider the ($d+1$)-frame graph with vertices $w^{(c)},  u^{(c)}_1,\dots,u^{(c)}_{d+1}, v^{(c)}_1,\dots,v^{(c)}_{d+1}$, where
$w^{(c)}$ is the root vertex, and the central path is $u^{(c)}_{d+1},\dots,u^{(c)}_1,v^{(c)}_1,\dots,v^{(c)}_{d+1}$.
Let $G_d^{(c)}$ be the graph obtained by inserting a copy of $G^{(c-1)}_d$ into each of the regions $\triangle{w^{(c)}v^{(c)}_kv^{(c)}_{k+1}}$ (the root vertex of $G^{(c-1)}_d$ is placed on $v^{(c)}_{k+1}$)
and $\triangle{v^{(c)}_{k+1}u^{(c)}_{k+1}v^{(c)}_k}$ (the root vertex of $G^{(c-1)}_d$ is placed on $u^{(c)}_{k+1}$)
for $k=1,\dots,d-1$ (see Figure~\ref{fig:g_d_c}).
We call $w^{(c)}$ the \emph{root vertex} of $G_d^{(c)}$.
The maximum degree of $G^{(c)}_d$ is at most $4d+13$.

Finally, we construct graph $H^{(c)}_{d}$  by assembling three copies of $G^{(c)}_d$ as in Figure~\ref{fig:h_d_c} (so that the angle $\angle{u_d^{(c-1)}w^{(c-1)}v_d^{(c-1)}}$ in  $G^{(c-1)}_d$ corresponds to the angles $\angle{s_1s_3s_4}$, $\angle{s_2s_1s_4}$, and $\angle{s_3s_2s_4}$)
and graph $\widetilde{H}^{(c)}_{d}$ by assembling three copies of $H^{(c)}_d$ as in Figure~\ref{fig:h_d_c_2}.
Since $\widetilde{H}^{(c)}_{d}$ is a planar 3-tree, the embedding of  $\widetilde{H}^{(c)}_{d}$ is unique.
The maximum degree of $\widetilde{H}^{(c)}_{d}$ is at most $8d+31$.

\begin{figure}[h]
\begin{center}
\includegraphics[scale=0.3, bb =  0 0 645 429,clip]{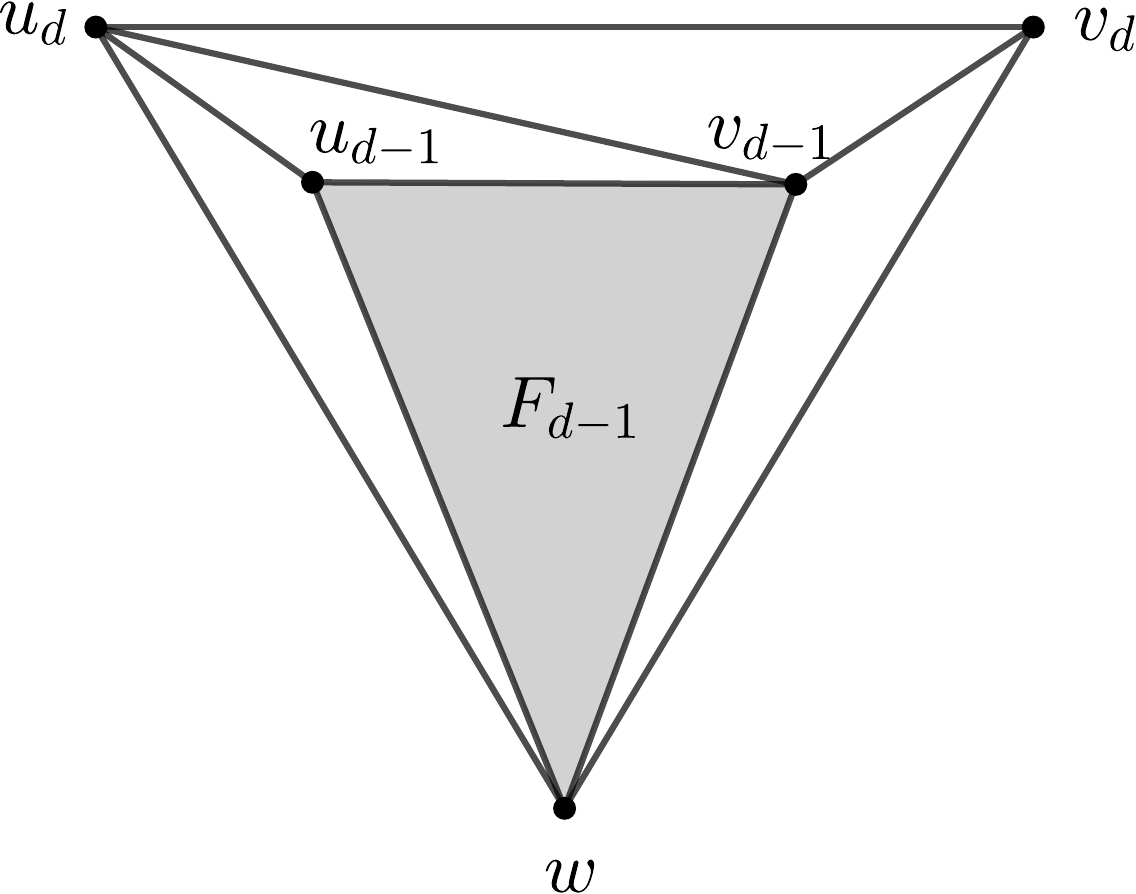}
\includegraphics[scale=0.25, bb = 0 0 790 497,clip]{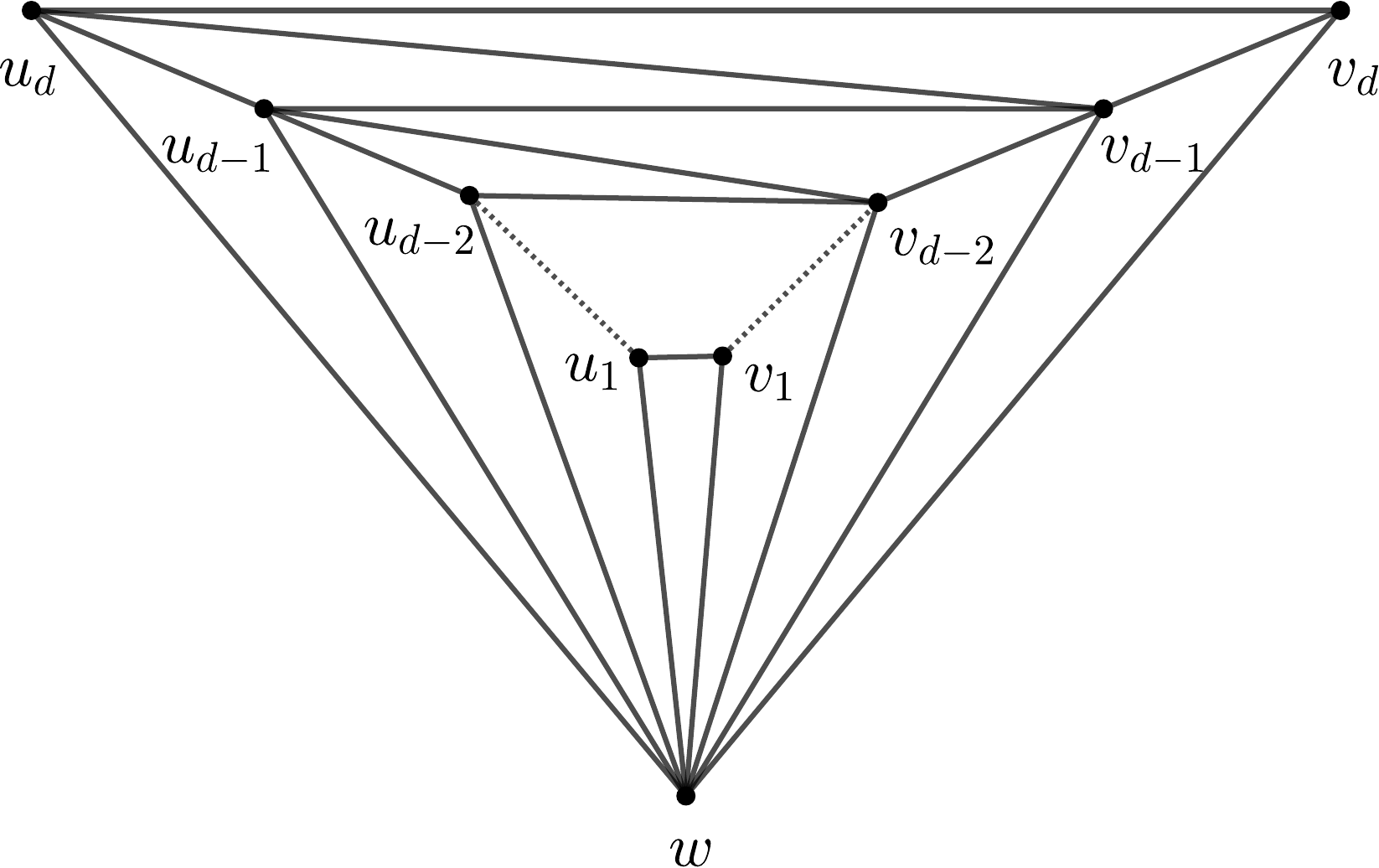}
\end{center}
\caption{Graph $F_d$}
\label{fig:f_d}
\end{figure}

\begin{figure}[h]
\begin{center}
\includegraphics[scale=0.3, bb =  0 0 1409 964,clip]{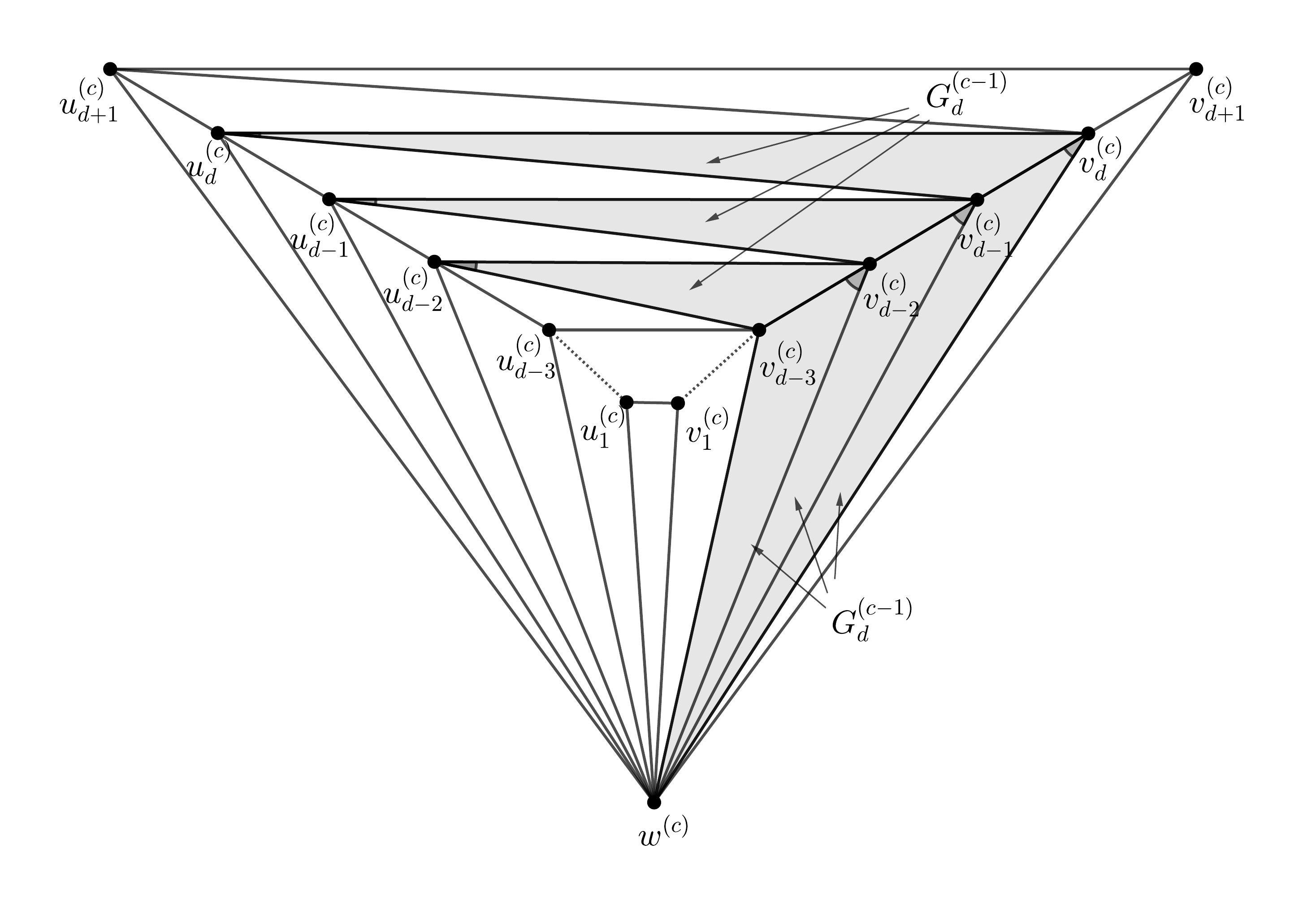}
\end{center}
\caption{Graph $G^{(c)}_{d}$}
\label{fig:g_d_c}
\end{figure}

\begin{figure}[h]
\begin{minipage}{0.5\linewidth}
\centering
\includegraphics[scale=0.18, bb =  0 0 708 556,clip]{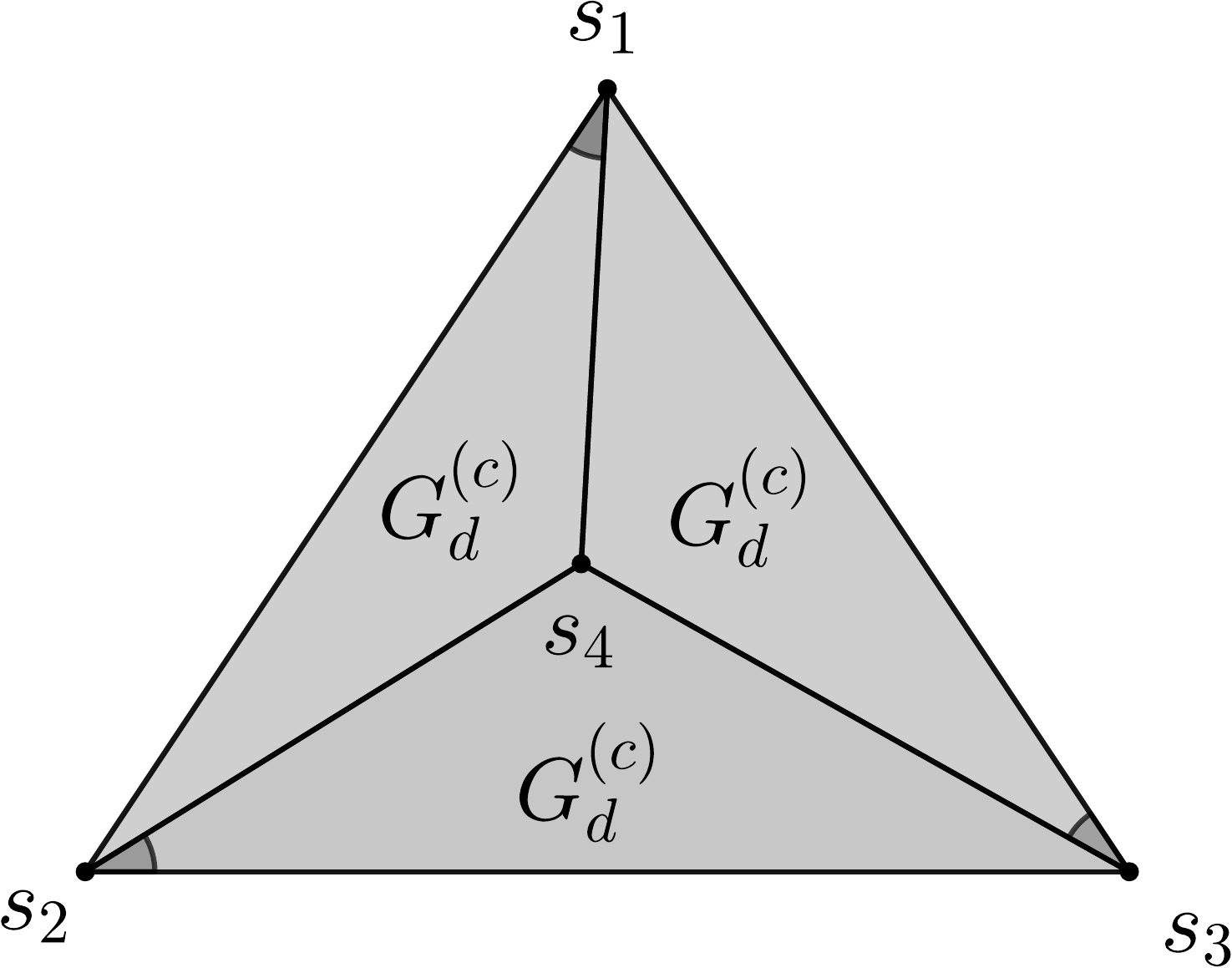}
\caption{Graph $H^{(c)}_d$}
\label{fig:h_d_c}
\end{minipage}
\begin{minipage}{0.5\linewidth}
\centering
\includegraphics[scale=0.18, bb =  0 0 691 570,clip]{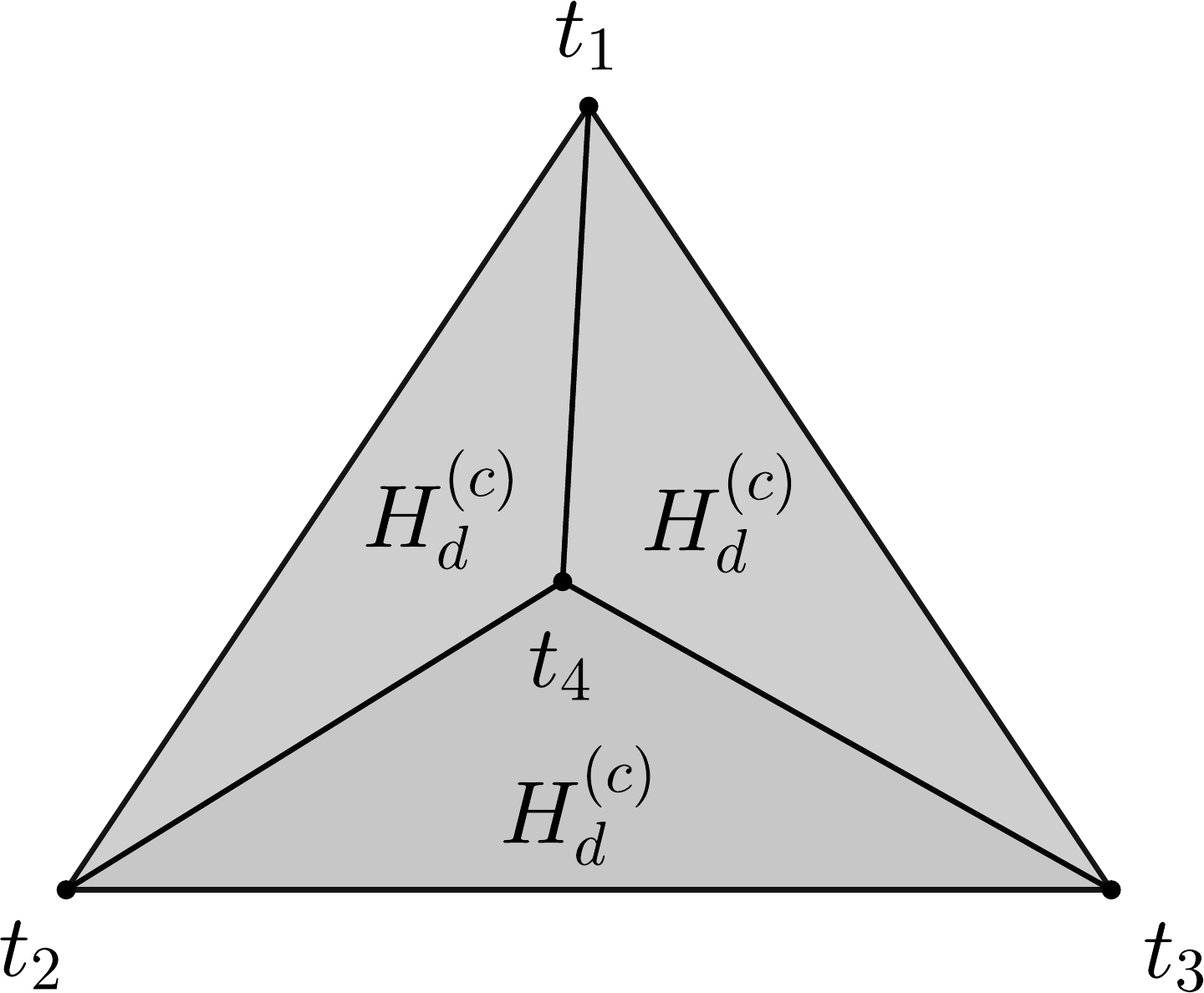}
\caption{Graph $\widetilde{H}^{(c)}_{d}$}
\label{fig:h_d_c_2}
\end{minipage}
\end{figure}

One can easily construct a drawing of graph $\widetilde{H}^{(1)}_{d}$ with angular resolution $\Omega (1/d)$.
On the other hand, one can prove that any drawing $\widetilde{H}^{(2)}_{d}$ has angular resolution $O((\log d)^{\frac{1}{2}}/d^{\frac{3}{2}})$ by a similar argument to the proof of Theorem~2 in $\cite{GT94}$.
In the following, we investigate the angular resolution of $\widetilde{H}^{(c)}_{d}$ for $c \geq 3$.

\subsection{Angular resolution of $\widetilde{H}^{(3)}_{d}$}
As a first step, we investigate the angular resolution of $\widetilde{H}^{(3)}_{d}$.
\begin{thm}
Graph $\widetilde{H}^{(3)}_d$ has an angle of size $O((\log d)^{\frac{1}{6}}/d^{\frac{3}{2}})$ in any planar straight-line drawing.
\label{thm:pre}
\end{thm}
\begin{proof}
The idea of the proof is similar to that of Theorem~2 in \cite{GT94}.
We first prove that graph $H^{(3)}_{d}$ has an angle of size $O((\log d)^{\frac{1}{6}}/d^{\frac{3}{2}})$ in any drawing where $\triangle{s_1s_2s_3}$ is the outer face.
On the contrary, let us assume that there is a drawing $\Gamma^{(3)}_{H}$ of $H^{(3)}_{d}$ in which $\triangle{s_1s_2s_3}$ is the outer face, and all angles have size $\omega ((\log d)^{\frac{1}{6}}/d^{\frac{3}{2}})$.
Without loss of generality, we assume that $\angle{s_1s_2s_3} \leq \pi /2$ in $\Gamma^{(3)}_{H}$.
Restricting $\Gamma^{(3)}_{H}$ to the region $\triangle{s_2s_3s_4}$, we obtain a drawing $\Gamma^{(3)}_{G}$ of $G^{(3)}_{3}$ with $\angle{u^{(3)}_{d+1}w^{(3)}v^{(3)}_{d+1}} \leq \pi /2$.
For $k=2,\dots,d$, let $\alpha_{k,1} \coloneqq \angle{v^{(3)}_{k-1}w^{(3)}v^{(3)}_k}$, $\alpha_{k,2} \coloneqq \angle{u^{(3)}_kw^{(3)}v^{(3)}_{k-1}}$, and $\alpha_{k,3} \coloneqq \angle{v^{(3)}_{1}w^{(3)}v^{(3)}_{k-1}}$, and
define $r_k \coloneqq \alpha_{k,3}/\alpha_{k,1}$.
Then, we have $\alpha_{2,1} = \left(\prod_{k=2}^d{\frac{r_k}{1+r_k}}\right)\angle{v^{(3)}_1w^{(3)}v^{(3)}_d}$.
Since $\alpha_{2,1} = \omega ((\log d)^{\frac{1}{6}}/d^{\frac{3}{2}})$ and $\angle{v^{(3)}_1w^{(3)}v^{(3)}_d} < \pi$, we have 
\[ \prod_{k=2}^d{\frac{r_k}{1+r_k}} = \prod_{k=2}^d{\frac{1}{\frac{1}{r_k}+1}} = \omega \left(\frac{(\log d)^{1/6}}{d^{3/2}}\right).\]
Let $i \coloneqq \arg\max_{k=2,\dots,d} r_k$. Since $r_i \geq 1$, we obtain $(1+1/r_i)^d = O(d^{\frac{3}{2}}/(\log d)^{\frac{1}{6}})$ from the above relation.
This implies $(1+1/r_i)^d = O(d^{3/2})$. Then, we obtain $r_i = \Omega (d/\log d)$ with simple manipulations.
Since $\alpha_{i,2}/\alpha_{i,1} \geq \alpha_{i,3}/\alpha_{i,1}$, we have  $\alpha_{i,2}/\alpha_{i,1} = \Omega (d/\log d)$.
Let $\beta_{i,1} \coloneqq \angle{v^{(3)}_{i-1}v^{(3)}_iu^{(3)}_i}$, $\beta_{i,2} \coloneqq \angle{w^{(3)}v^{(3)}_iv^{(3)}_{i-1}}$, $\gamma_{i,1} \coloneqq \angle{v^{(3)}_{i-1}u^{(3)}_iw^{(3)}}$, $\gamma_{i,2} \coloneqq \angle{v^{(3)}_{i-1}u^{(3)}_iv^{(3)}_{i}}$.
By Lemma~\ref{lem:geometric}, we have either $\beta_{i,2}/\beta_{i,1} = O((\log d/d)^{1/2})$ or $\gamma_{i,2}/\gamma_{i,1}= O((\log d/d)^{1/2})$, which 
implies $\beta_{i,2} = O((\log d/d)^{1/2})$ or $\gamma_{i,2} = O((\log d/d)^{1/2})$.
Restricting the drawing $\Gamma^{(3)}_{G}$ to the region $\triangle{w^{(3)}v^{(3)}_iv_{i-1}^{(3)}}$ or  $\triangle{v^{(3)}_{i-1}u^{(3)}_iv^{(3)}_{i}}$,  we obtain a drawing $\Gamma^{(2)}_{G}$ of $G^{(2)}_{d}$ with $\angle{u^{(2)}_{d+1}w^{(2)}v^{(2)}_{d+1}} =  O((\log d/d)^{1/2})$.

In $\Gamma^{(2)}_{G}$, we consider $j \coloneqq \arg\min_{k=\frac{d}{2},\dots,d} \angle{v^{(2)}_kw^{(2)}v^{(2)}_{k-1}}$,
 $\alpha'_{1} \coloneqq \angle{v^{(2)}_jw^{(2)}v^{(2)}_{j-1}}$, and $\alpha'_{2} \coloneqq \angle{u^{(2)}_jw^{(2)}v^{(2)}_{j-1}}$.
Then, we have $\alpha'_{1}  \leq \frac{2}{d+2} \times \angle{u^{(2)}_dw^{(2)}v^{(2)}_d} = O((\log d)^{\frac{1}{2}}/d^{\frac{3}{2}})$.
On the other hand, we have $\alpha'_{2}  = \omega ((\log d)^{\frac{1}{6}}/d^{\frac{1}{2}})$ since $\alpha'_2$ contain more than $d-1$ angles of size  $\omega ((\log d)^{\frac{1}{6}}/d^{\frac{3}{2}})$ in its interior.
Combining these two relations, we obtain 
$\alpha'_{2}  /\alpha'_{1} =\Omega(d/(\log d)^{\frac{1}{3}})$.
Letting $\beta'_{1} \coloneqq \angle{u^{(2)}_jv^{(2)}_jv^{(2)}_{j-1}}$, $\beta'_{2} \coloneqq \angle{w^{(2)}v^{(2)}_jv^{(2)}_{j-1}}$, $\gamma'_{1} \coloneqq \angle{w^{(2)}u^{(2)}_jv^{(2)}_{j-1}}$,  and $\gamma'_{2} \coloneqq \angle{v^{(2)}_ju^{(2)}_jv^{(2)}_{j-1}}$, 
we have either $\beta'_{2} /\beta'_1 =O((\log d)^{\frac{1}{6}}/d^{\frac{1}{2}})$ or $\gamma'_2 /\gamma'_1 =O((\log d)^{\frac{1}{6}}/d^{\frac{1}{2}})$ by Lemma~\ref{lem:geometric}.
This implies $\beta'_{2} =O((\log d)^{\frac{1}{6}}/d^{\frac{1}{2}})$ or $\gamma'_2=O((\log d)^{\frac{1}{6}}/d^{\frac{1}{2}})$.
Since both of  the angles $\angle{w^{(2)}v^{(2)}_jv^{(2)}_{j-1}}$ and $\angle{v^{(2)}_ju^{(2)}_jv^{(2)}_{j-1}}$ include $2d$ angles in their interior, there is an angle of size $O((\log d)^{\frac{1}{6}}/d^{\frac{3}{2}})$ in $\Gamma^{(2)}_{G}$. This is a contradiction.
Therefore, graph $H^{(3)}_{d}$ must have an angle of size $O((\log d)^{\frac{1}{6}}/d^{\frac{3}{2}})$ in any drawing where $\triangle{s_1s_2s_3}$ is the outer face.

Now we can readily prove the theorem.
Take a drawing of  $\widetilde{H}^{(3)}_d$ arbitrarily.
Since any drawing of  $\widetilde{H}^{(3)}_d$  contains a  drawing of $H^{(3)}_d$ where $\triangle{s_1s_2s_3}$ is the outer face,
it must contain an angle of size $O((\log d)^{\frac{1}{6}}/d^{\frac{3}{2}})$ by the above discussion.
Hence the theorem is proved.
\end{proof}

\subsection{Angular resolution of $\widetilde{H}^{(c)}_{d}$}
Now we investigate the angular resolution of $\widetilde{H}^{(c)}_{d}$.
Extending the proof of Theorem~\ref{thm:pre}, we prove the following theorem:
\begin{thm}
Let $c \geq 2$. In any planar straight-line drawing of graph $\widetilde{H}^{(c)}_d$, there is an angle of size $O((\log d)^{\frac{1}{2\cdot 3^{c-2}}}/d^{\frac{3}{2}})$.
\label{thm:main}
\end{thm}
To prove this theorem, we first show the following claim:
\begin{cl}
Let $a$ be any constant, and $\Gamma$ be a drawing  of  the $d$-frame graph $F_d$ such that $\angle{u_{d}wv_{d}} = O((\log d)^{\frac{1}{2 \cdot 3^{a-1}}}/d^{\frac{1}{2}})$ and all angles have size $\omega((\log d)^{\frac{1}{2 \cdot 3^a}}/d^{\frac{3}{2}})$.
There exists an index $j$ such that
$\angle{wv_jv_{j-1}}=O((\log d)^{\frac{1}{2\cdot 3^{a}}}/d^{\frac{1}{2}})$ or $\angle{v_ju_jv_{j-1}}=O((\log d)^{\frac{1}{2\cdot 3^{a}}}/d^{\frac{1}{2}})$ holds in $\Gamma$.
\label{cl}
\end{cl}
\begin{proof}
Let $j \coloneqq \arg\min_{k=\frac{d}{2},\dots,d} \angle{v_kwv_{k-1}}$,
 $\alpha_{1} \coloneqq \angle{v_jwv_{j-1}}$, and $\alpha_{2} \coloneqq \angle{u_jwv_{j-1}}$.
Applying a similar argument to the proof of Theorem~\ref{thm:pre}, we obtain $\alpha_{1}   = O((\log d)^{\frac{1}{2\cdot 3^{a-1}}}/d^{\frac{3}{2}})$ and $\alpha_2  = \omega((\log d)^{\frac{1}{2\cdot 3^a}}/d^{\frac{1}{2}})$.
This leads to that
$\alpha_{2}  /\alpha_{1} =\Omega(d/(\log d)^{\frac{1}{3^{a}}})$.
Let $\beta_{1} \coloneqq \angle{u_jv_jv_{j-1}}$, $\beta_{2} \coloneqq \angle{wv_jv_{j-1}}$, $\gamma_{1} \coloneqq \angle{wu_jv_{j-1}}$, and $\gamma_{2} \coloneqq \angle{v_ju_jv_{j-1}}$. 
By Lemma~\ref{lem:geometric}, we have either $\beta_{2} /\beta_{1} =O((\log d)^{\frac{1}{2\cdot 3^{a}}}/d^{\frac{1}{2}})$ or $\gamma_{2} /\gamma_{1} =O((\log d)^{\frac{1}{2\cdot 3^{a}}}/d^{\frac{1}{2}})$.
This implies $\angle{wv_jv_{j-1}}=O((\log d)^{\frac{1}{2\cdot 3^{a}}}/d^{\frac{1}{2}})$ or $\angle{v_ju_jv_{j-1}}=O((\log d)^{\frac{1}{2\cdot 3^{a}}}/d^{\frac{1}{2}})$.
\end{proof}
\\
\\
Now we are ready to prove Theorem~\ref{thm:main}.
\begin{proof_main}
Let us assume that there exists a drawing  of $H^{(c)}_d$ such that $\triangle{s_1s_2s_3}$ is the outer face, and
all angles have size $\omega((\log d)^{\frac{1}{2\cdot 3^{c-2}}}/d^{\frac{3}{2}})$.
Using the same argument as in the proof of Theorem~\ref{thm:pre}, we obtain a drawing of $G^{(c-1)}_d$ such that $\angle{u^{(c-1)}_{d+1}w^{(c-1)}v^{(c-1)}_{d+1}} = O((\log d)^{\frac{1}{2}}/d^{\frac{1}{2}})$
and all angles have size $\omega((\log d)^{\frac{1}{2\cdot 3^{c-2}}}/d^{\frac{3}{2}})$.
Applying Claim~\ref{cl} ($c-2$) times, we obtain a drawing of $G^{(1)}_d$ such that $\angle{u_{d+1}^{(1)}w^{(1)}v_{d+1}^{(1)}}=O((\log d)^{\frac{1}{2\cdot 3^{c-2}}}/d^{\frac{1}{2}})$, and all angles have size $\omega((\log d)^{\frac{1}{2\cdot 3^{c-2}}}/d^{\frac{3}{2}})$.
Since the angle $\angle{u_{d+1}^{(1)}w^{(1)}v_{d+1}^{(1)}}$ include $2d$ angles in their interior, there is an angle of size $O((\log d)^{\frac{1}{2\cdot 3^{c-2}}}/d^{\frac{3}{2}})$. This is a contradiction.
Therefore,  graph $H^{(c)}_d$ has an angle of size $O((\log d)^{\frac{1}{2\cdot 3^{c-2}}}/d^{\frac{3}{2}})$ in any drawing where $\triangle{s_1s_2s_3}$ is the outer face.
Since any drawing of  $\widetilde{H}^{(c)}_d$  contains a  drawing of $G^{(c)}_d$ in which $\triangle{s_1s_2s_3}$ is the outer face,
it must contain an angle of size  $O((\log d)^{\frac{1}{2\cdot 3^{c-2}}}/d^{\frac{3}{2}})$.
Hence the theorem is proved.
\end{proof_main}
\\
\\
By considering the case $c=2-\lfloor \log_3 (2\varepsilon) \rfloor$ for $\varepsilon > 0$, we obtain the following corollary.
\begin{cor}
Let $\varepsilon$ be any positive constant.
There is a family of planar graphs with maximum degree $d$ that have angular resolution $O((\log d)^{\varepsilon}/d^{\frac{3}{2}})$ in any planar straight-line drawing.
\end{cor}

\subsection*{Acknowledgements}
The author is thankful to Masaya Sekine for helpful discussions.
This work was supported by JSPS KAKENHI Grant Number JP19K20210.

\end{document}